\documentclass{article}
\usepackage{spconf,amsmath,amssymb,dsfont,graphicx,verbatim}
\graphicspath{{figures/}}
\usepackage{amsmath,amssymb,dsfont,verbatim}
\usepackage{import}
\usepackage{subcaption}
\usepackage{url, hyperref}
\usepackage{cite}
\usepackage{xcolor}
\usepackage{booktabs} 
\usepackage{tikz}
\usetikzlibrary{mindmap}

\allowdisplaybreaks[0]

\let\oldbibliography\thebibliography
\renewcommand{\thebibliography}[1]{\oldbibliography{#1}
\setlength{\itemsep}{-2pt}} %Reducing spacing in the bibliography.

\DeclareMathOperator{\bcw}{{\boldsymbol{\scriptstyle\mathcal{W}}}}
\usepackage{amsthm}

\DeclareMathOperator{\T}{\mathsf{T}}
\DeclareMathOperator{\E}{\mathds{E}}

\DeclareMathOperator{\w}{\boldsymbol{w}}
\DeclareMathOperator{\x}{\boldsymbol{x}}
\DeclareMathOperator{\s}{\boldsymbol{s}}

\DeclareMathOperator{\h}{\boldsymbol{h}}

\usepackage{amsthm}

\theoremstyle{plain}

\newtheorem{definition}{Definition}
\newtheorem{assumption}{Assumption}
\newtheorem{theorem}{Theorem}

\newtheorem{lemma}{Lemma}

\title{Linear Speedup in Saddle-Point Escape for Decentralized Non-Convex Optimization}
\name{Stefan Vlaski and Ali H. Sayed\thanks{This work was supported in part by NSF grant CCF-1524250. Emails:\{stefan.vlaski, ali.sayed\}@epfl.ch.}
\address{School of Engineering, \'{E}cole Polytechnique F\'{e}d\'{e}rale de Lausanne}}
\begin{document}
\ninept
\maketitle
\begin{abstract}
Under appropriate cooperation protocols and parameter choices, fully decentralized solutions for stochastic optimization have been shown to match the performance of centralized solutions and result in linear speedup (in the number of agents) relative to non-cooperative approaches in the strongly-convex setting. More recently, these results have been extended to the pursuit of first-order stationary points in non-convex environments. In this work, we examine in detail the dependence of second-order convergence guarantees on the spectral properties of the combination policy for non-convex multi agent optimization. We establish linear speedup in saddle-point escape time in the number of agents for symmetric combination policies and study the potential for further improvement by employing asymmetric combination weights. The results imply that a linear speedup can be expected in the pursuit of \emph{second-order stationary} points, which exclude local maxima as well as strict saddle-points and correspond to local or even global minima in many important learning settings.
\end{abstract}
\section{Introduction and Related Work}
\label{sec:intro}
We consider a collection of \( K \) agents, where each agent \( k \) is equipped with a local stochastic cost function:
\begin{equation}
  J_k(w) \triangleq \E Q_k( w; \x_{k} )
\end{equation}
where \( w \in \mathds{R}^M \) denotes a parameter vector and \( \x_{k} \) denotes the random data at agent \( k \). We construct the global cost function:
\begin{equation}\label{eq:global_problem}
  J(w) \triangleq \sum_{k=1}^K p_k J_k(w)
\end{equation}
where the \( p_k \ge 0 \) denote convex combination weights that add up to one, i.e. \( \sum_{k=1}^K p_k = 1 \). When data realizations for \( \x_{k} \) can be aggregated at a central location, descent along the negative gradient of~\eqref{eq:global_problem} can be approximated by means of a \emph{centralized} stochastic gradient algorithm of the form~\cite{Polyak87,Sayed14}:
\begin{equation}\label{eq:centralized}
  \w_i^{\mathrm{cent}} = \w_{i-1} - \mu \widehat{\nabla J}(\w_{i-1}^{\mathrm{cent}})
\end{equation}
where \( \widehat{\nabla J}(\cdot) \) denotes a stochastic gradient approximation constructed at time \( i-1 \). One possible construction is to let:
\begin{equation}\label{eq:k_grad}
  \widehat{\nabla J}^{c, K}(\w_{i-1}^{\mathrm{cent}}) \triangleq \sum_{k=1}^K p_k \nabla Q_k( \w_{i-1}^{\mathrm{cent}}; \x_{k,i-1} )
\end{equation}
which is obtained by employing a weighted combination of instantaneous approximations using all \( K \) realizations available at time \( i-1 \). This construction requires the evaluation of \( K \) (stochastic) gradients per iteration. If computational constraints limit the number of gradient evaluations per iteration to one, we can instead randomly sample an agent location \( k \) from the available data and let:
\begin{equation}\label{eq:1_grad}
  \widehat{\nabla J}^{c, 1}(\w_{i-1}^{\mathrm{cent}}) = \nabla Q_k( \w_{i-1}^{\mathrm{cent}}; \x_{k,i-1} ), \ \mathrm{with}\ \mathrm{prob.}\ p_k,
\end{equation}
The evident drawback of such a simplified centralized strategy is that only one sample is processed and a large number of samples is discarded at every iteration. We can hence expect the construction~\eqref{eq:k_grad} to result in better performance relative to the simplified choice~\eqref{eq:1_grad}. When communication constraints limit the exchange of information among agents, we can instead appeal to decentralized strategies. For the purpose of this work, we shall focus on the standard diffusion strategy, which takes the form:
\begin{subequations}
\begin{align}
  \boldsymbol{\phi}_{k,i} &= \w_{k,i-1} - \mu \nabla Q_k( \w_{k, i-1}; \x_{k,i-1} ) \label{eq:adapt}\\
  \w_{k,i} &= \sum_{\ell=1}^{N} a_{\ell k} \boldsymbol{\phi}_{\ell,i}\label{eq:combine}
\end{align}
\end{subequations}
where \( a_{\ell k} \) denote convex combination coefficients satisfying:
\begin{equation}\label{eq:combinationcoef}
    a_{\ell k} \geq 0, \quad \sum_{\ell \in \mathcal{N}_k} a_{\ell k}=1, \quad a_{\ell k} = 0\ \mathrm{if}\ \ell \notin \mathcal{N}_k
  \end{equation}
The symbol \( \mathcal{N}_k \) denotes the set of neighbors of agent \( k \). When the graph is strongly-connected, it follows from the Perron-Frobenius theorem that the combination matrix \( A \) has a spectral radius of one and a single eigenvalue at one with corresponding eigenvector~\cite{Sayed14}:
\begin{equation}\label{eq:perron}
  Ap=p, \quad \mathds{1}^{\T} p=1, \quad p_k > 0
\end{equation}
Comparing the diffusion strategy~\eqref{eq:adapt}--\eqref{eq:combine} to the centralized constructions~\eqref{eq:k_grad} or~\eqref{eq:1_grad}, we observe that the adaptation step~\eqref{eq:adapt} carries the same complexity \emph{per agent} as the simplified construction~\eqref{eq:1_grad}. However, since these computations are performed at \( K \) agents in parallel, and the information is diffused over the network through the combination step~\eqref{eq:combine}, we expect the diffusion strategy~\eqref{eq:adapt}--\eqref{eq:combine} to outperform the simplified centralized strategy~\eqref{eq:1_grad} and more closely match the full construction~\eqref{eq:k_grad}. In fact, the spectral properties~\eqref{eq:perron} of the combination weights~\eqref{eq:combinationcoef} allow us to establish the following relation for the weighted network mean \( \w_{c, i} \triangleq \sum_{k=1}^K p_k \w_{k,i} \)~\cite{Chen15transient, Sayed14}:
\begin{equation}\label{eq:centroid_recursion}
  \w_{c,i} = \w_{c,i-1} - \mu \sum_{k=1}^K p_k \nabla Q_k( \w_{k, i-1}; \x_{k,i-1} )
\end{equation}
which \emph{almost} corresponds to the centralized recursion~\eqref{eq:centralized}--\eqref{eq:k_grad} with the full gradient approximation~\eqref{eq:k_grad} with the only difference being that the stochastic gradients are evaluated at the individual iterates \( \w_{k, i-1} \) instead of the weighted network centroid \(\w_{c, i-1} \). So long as the iterates \( \w_{k, i-1} \) cluster around the network centroid, and under appropriate smoothness conditions on the (stochastic) gradients, it is hence to be expected that the network centroid~\eqref{eq:centroid_recursion} will match the performance of the full gradient approximation~\eqref{eq:k_grad}. This intuition has been studied in great detail and formalized for strongly convex cost functions, establishing that \emph{all} iterates \( \w_{k, i} \) in~\eqref{eq:adapt}--\eqref{eq:combine} will actually match the centralized full gradient approximation~\eqref{eq:k_grad} both in terms of convergence rate~\cite{Chen15transient} and steady-state error~\cite{Chen15performance}, which implies a linear improvement over the simplified construction~\eqref{eq:1_grad} in terms of the number of agents~\cite{Sayed14} when employing a symmetric combination policy for which \( p_k = \frac{1}{K} \).

More recently, these results have been extended to the pursuit of first-order stationary points in non-convex environments~\cite{Lian17, Tang18} for consensus and the exact diffusion algorithm~\cite{Yuan19}. First-order stationary points can include saddle-points and even local maxima and can generate a bottleneck for many optimization algorithms and problem formulations~\cite{Du17}. Hence, the purpose of this work is is to establish that linear speedup can also be expected in the escape from saddle-points and pursuit of second-order stationary points for non-convex optimization problems. To this end, we refine and exploit recent results in~\cite{Vlaski19nonconvexP1, Vlaski19nonconvexP2}.

\subsection{Related Works}
\label{sec:related}
Strategies for decentralized optimization include incremental strategies~\cite{Bertsekas97incremental}, and decentralized gradient descent (or consensus)~\cite{Ram10}, as well as the diffusion algorithm~\cite{Sayed14, Chen15transient, Vlaski19smoothing}. A second class of strategies is based on primal-dual arguments~\cite{Jakovetic11, Duchi12, Jaggi14, Tsianos12, Shi15extra, Yuan19}. While most of these algorithms are applicable to non-convex optimization problems, most performance guarantees in non-convex environments are limited to establishing convergence to first-order stationary points, i.e., points where the gradient is equal to zero~\cite{Lorenzo16, Tatarenko17,Lian17, Tang18, Wang18}.

Landscape analysis of commonly employed loss surfaces has uncovered that in many important settings such as tensor decomposition~\cite{Ge15}, matrix completion~\cite{Ge16}, low-rank recovery~\cite{Ge17}, as well as certain deep learning architectures~\cite{Kawaguchi16}, \emph{all} local minima correspond to \emph{global} minima and \emph{all} other first-order stationary points have a \emph{strict-saddle} property, which states that the Hessian matrix has at least one negative eigenvalue. These results have two implications. First, while first-order stationarity is a useful result in the sense that it ensures stability of the algorithm, even in non-convex environments, it is not sufficient to guarantee satisfactory performance, since first-order stationary points include strict saddle-points, which need not be globally or even locally optimal. On the other hand, establishing the escape from strict saddle-points, is sufficient to establish convergence to \emph{global} optimality in all of these problems.

These observations have sparked a number of works examining second-order guarantees of local descent algorithms. Strategies for the escape from saddle-points can generally be divided into one of two classes. First, since the Hessian at every strict-saddle point, by definition, contains at least one negative eigenvalue, the descent direction can be identified by directly employing the Hessian matrix~\cite{Nesterov06} or through an intermediate search for the negative curvature direction~\cite{Fang18, Allen18neon}. The second class of strategies leverages the fact that perturbations in the initialization~\cite{Lee16} or the update direction~\cite{Ge15, Jin19, HadiDaneshmand18, Vlaski19single} cause iterates of first-order algorithms to not get ``stuck'' in strict saddle-points, which can be shown to be unstable. Recently these results have been extended to decentralized optimization with deterministic gradients and random initialization~\cite{Scutari18} as well as stochastic gradients with diminishing step-size and decaying additive noise~\cite{Swenson19} as well as constant step-sizes~\cite{Vlaski19nonconvexP1, Vlaski19nonconvexP2}. We establish in this work, that the saddle-point escape time of the diffusion strategy~\eqref{eq:adapt}--\eqref{eq:combine} decays linearly with the number of agents in the network when symmetric combination policies are employed and show how asymmetric combination policies can result in further improvement when agents have access to estimates of varying quality.

\section{Modeling Conditions}
\label{sec:modeling}
We shall be employing the following common modeling conditions~\cite{Sayed14, Ge15, HadiDaneshmand18, Swenson19}. See~\cite{Vlaski19nonconvexP1, Vlaski19nonconvexP2} for a discussion.
\begin{assumption}[\textbf{Smoothness}]\label{as:lipschitz}
  For each \( k \), the gradient \( \nabla J_k(\cdot) \) is Lipschitz, namely, for any \( x,y \in \mathds{R}^{M} \):
  \begin{equation}\label{eq:lipschitz}
    \|\nabla J_k(x) - \nabla J_k(y)\| \le \delta \|x-y\|
  \end{equation}
	Furthermore, \( J_k(\cdot) \) is twice-differentiable with Lipschitz Hessian:
  \begin{equation}
    {\| \nabla^2 J_k(x) - \nabla^2 J_k(y) \|} \le \rho \|x - y\|
  \end{equation}
	For each pair of agents \( k \) and \( \ell \), the gradient disagreement is bounded, namely, for any \( x \in \mathds{R}^{M} \):
  \begin{equation}\label{eq:bounded}
    \|\nabla J_k(x) - \nabla J_{\ell}(x)\| \le G
  \end{equation}
\end{assumption}\hfill\qed
\begin{assumption}[\textbf{Gradient noise process}]\label{as:gradientnoise}
  For each \( k \), the gradient noise process is defined as
  \begin{equation}
    \s_{k,i}(\w_{k,i-1}) = \widehat{\nabla J}_k(\w_{k,i-1}) - \nabla J_k(\w_{k,i-1})
  \end{equation}
  and satisfies
  \begin{subequations}
    \begin{align}
      \E \left\{ \s_{k,i}(\w_{k,i-1}) | \boldsymbol{\mathcal{F}}_{i-1} \right\} &= 0 \label{eq:conditional_zero_mean}\\
      \E \left\{ \|\s_{k,i}(\w_{k,i-1})\|^4 | \boldsymbol{\mathcal{F}}_{i-1} \right\} &\le \sigma_k^4 \label{eq:gradientnoise_fourth}
    \end{align}
  \end{subequations}
  where we denote by \( \boldsymbol{\mathcal{F}}_{i} \) the filtration generated by the random processes \( \w_{k, j} \) for all \( k \) and \( j \le i \) and for some non-negative constants \( \sigma_k^4 \). We also assume that the gradient noise processes are pairwise uncorrelated over the space conditioned on \( \boldsymbol{\mathcal{F}}_{i-1} \).\hfill\qed%
\end{assumption}
\begin{assumption}[\textbf{Lipschitz covariances}]\label{as:lipschitz_covariance}
  The gradient noise process has a Lipschitz covariance matrix, i.e.,
  \begin{equation}
    R_{s, k}(\w_{k, i-1}) \triangleq \E \left \{ \s_{k, i}(\w_{k, i-1}) {\s_{k, i}(\w_{k, i-1})}^{\T} | \boldsymbol{\mathcal{F}}_{i-1}\right \}
  \end{equation}
  satisfies
  \begin{equation}\label{eq:lipschitz_r}
    \| R_{s, k}(x) - R_{s, k}(y) \| \le \beta_R {\| x - y \|}^{\gamma}
  \end{equation}
  for some \( \beta_R \) and \( 0 < \gamma \le 4\).\hfill\qed
\end{assumption}
\noindent We shall also make the simplifying assumption.
\begin{assumption}[\textbf{Gradient noise lower bound}]\label{as:noise_in_saddle}
  The gradient noise covariance \( R_{s, k}(x) \) at every agent is bounded from below:
  \begin{equation}
    R_{s, k}(x) \ge \sigma_{\ell, k} I
  \end{equation}\hfill\qed
\end{assumption}
\noindent This condition can be loosened significantly by requiring a gradient noise component to be present only in the vicinity of strict saddle-points and only in the local descent direction, see e.g.~\cite{HadiDaneshmand18, Vlaski19nonconvexP2}. Nevertheless, the simplified condition can always be ensured for example by adding a small amount of isotropic noise, similar to~\cite{Ge15, Jin19} and will be sufficient for the purpose this work.

\section{Convergence Analysis}
\subsection{Noise Variance Relations}
The performance guarantees established in~\cite{Vlaski19nonconvexP1, Vlaski19nonconvexP2} depend on the statistical properties of the weighted gradient noise term:
\begin{equation}\label{eq:average_noise}
  \s_{i} \triangleq \sum_{k=1}^K p_k \s_{k, i}(\w_{k, i-1})
\end{equation}
Under assumptions~\ref{as:lipschitz}--\ref{as:noise_in_saddle}, we can refine the bounds from~\cite{Vlaski19nonconvexP1}:
\begin{lemma}[\textbf{Variance Bounds}]\label{LEM:VARIANCE_BOUNDS}
	Under assumptions~\ref{as:lipschitz}--\ref{as:noise_in_saddle} we have:
	\begin{align}
		\E \left \{ {\|\boldsymbol{s}_{i}\|}^2 | \boldsymbol{\mathcal{F}}_{i-1} \right \} &\le \sum_{k=1}^K p_k^2 \sigma_k^2 \label{eq:variance_bound} \\
    \left( \sum_{k=1}^K p_k^2 \sigma_{k,\ell}^2 \right)I &\le \E \s_i \s_i^{\T} \le \left( \sum_{k=1}^K p_k^2 \sigma_k^2 \right)I\label{eq:hessian_bound}
	\end{align}
\end{lemma}
\begin{proof}
	Relations~\eqref{eq:variance_bound} and~\eqref{eq:hessian_bound} follow from the pairwise uncorrelatedness condition in assumption~\ref{as:gradientnoise} after cross-multiplying.
\end{proof}
From~\eqref{eq:variance_bound} we observe that the average noise term~\eqref{eq:average_noise} driving the network centroid experiences a variance reduction. Specifically, in the case when \( p_k = 1/K \) and \( \sigma_k = \sigma \) we would obtain \( \E \left \{ {\|\boldsymbol{s}_{i}\|}^2 | \boldsymbol{\mathcal{F}}_{i-1} \right \} \le {\sigma^2}/{K} \). This \( K \)-fold reduction in gradient noise variance is at the heart of the improved performance established for strongly-convex costs~\cite{Sayed14} and in the pursuit of first-order stationary points~\cite{Lian17}. We shall establish in the sequel that this improvement also holds in the time required to escape from undesired saddle-points.

\subsection{Space Decomposition}
\begin{definition}[Sets]\label{DEF:SETS}
  The parameter space \( \mathds{R}^M \) is decomposed into:
  \begin{align}
    \mathcal{G} &\triangleq \left \{ w : {\left \| \nabla J(w) \right \|}^2 \ge \mu \frac{c_2}{c_1}\left(1+ \frac{1}{\pi}\right) \right \} \label{eq:define_g}\\
    \mathcal{H} &\triangleq \left \{ w : w \in \mathcal{G}^C, \lambda_{\min}\left( \nabla^2 J(w) \right) \le -\tau \right \} \label{eq:define_h}\\
    \mathcal{M} &\triangleq \left \{ w : w \in \mathcal{G}^C, \lambda_{\min}\left( \nabla^2 J(w) \right) > -\tau \right \} \label{eq:define_m}
  \end{align}
  where \( \tau \) is a small positive parameter, \( 0 < \pi < 1 \) is a parameter to be chosen, \( c_1 = \frac{1}{2} \left(1 - 2 \mu \delta\right) = O(1) \) and \( c_2 =\frac{\delta}{2} \left(\sum_{k=1}^K p_k^2 \sigma_k^2\right) = O\left(\sum_{k=1}^K p_k^2 \sigma_k^2\right) \). Note that \( \mathcal{G}^C = \mathcal{H} \cup \mathcal{M} \). We also define the probabilities \(\pi^{\mathcal{G}}_i \triangleq \mathrm{Pr}\left \{ \w_{c, i} \in \mathcal{G} \right \}\), \( \pi^{\mathcal{H}}_i \triangleq \mathrm{Pr}\left \{ \w_{c, i} \in \mathcal{H} \right \}\) and \( \pi^{\mathcal{M}}_i \triangleq \mathrm{Pr}\left \{ \w_{c, i} \in \mathcal{M} \right \} \). Then for all \( i \), we have \( \pi^{\mathcal{G}}_i + \pi^{\mathcal{H}}_i + \pi^{\mathcal{M}}_i = 1 \). \hfill\qed
\end{definition}
\noindent Points in the complement of \( \mathcal{G} \) have small gradient norm and hence correspond to approximately first-order stationary points. These points are further classified into strict-saddle points \( \mathcal{H} \), where the Hessian has a significant negative eigenvalue, and second-order stationary points \( \mathcal{M} \). Pursuit of second-order stationary points requires descent for points in \( \mathcal{G} \) as well as \( \mathcal{H} \).

\subsection{Performance Guarantees}
Due to space limitations, we forego a detailed discussion on the derivation of the second-order guarantees of the diffusion algorithm~\eqref{eq:adapt}--\eqref{eq:combine} and refer the reader to~\cite{Vlaski19nonconvexP1, Vlaski19nonconvexP2}. We instead briefly list the guarantees resulting from the variance bounds~\eqref{eq:variance_bound}--\eqref{eq:hessian_bound} and will focus on the dependence on the combination policy further below. Adjusting the theorems in~\cite{Vlaski19nonconvexP1, Vlaski19nonconvexP2} to account for the variance bounds~\eqref{eq:variance_bound}--\eqref{eq:hessian_bound}, we obtain:
\begin{theorem}[\textbf{Network disagreement (4th order)}]\label{LEM:NETWORK_DISAGREEMENT_FOURTH}
	Under assumptions~\ref{as:lipschitz}~\ref{as:gradientnoise}, the network disagreement is bounded after sufficient iterations \( i \ge i_o \) by:
	\begin{align}
		&\:\E {\left \| \bcw_i - \left( \mathds{1} p^{\T} \otimes I \right) \bcw_{i} \right \|}^4 \notag \\
    \le&\: \mu^4  \frac{{\left \| \mathcal{V}_L \right \|}^4{\left \|J_{\epsilon}^{\T} \right \|}^4}{{\left(1-{\left \|J_{\epsilon}^{\T} \right \|}\right)}^4} {\| \mathcal{V}_R^{\T} \|}^4 K^2 \left( G^4 + \max_k \sigma_k^4 \right) + o(\mu^4)\label{eq:network_disagreement_fourth}
	\end{align}
  where \( {\left \|J_{\epsilon}^{\T} \right \|} = \lambda_2(A) + \epsilon \approx \lambda_2(A) \) denotes the mixing rate of the adjacency matrix, \( \mathcal{A} = \mathcal{V}_{\epsilon} \mathcal{J} \mathcal{V}_{\epsilon}^{-1} \) with \( \mathcal{V}_{\epsilon} = \mathrm{row}\left\{ p \otimes I, \mathcal{V}_{R} \right\} \) and \( \mathcal{V}_{\epsilon}^{-1} = \mathrm{col}\left\{ \mathds{1}^{\T}, \mathcal{V}_{L}^{\T} \right\} \), \( i_o = {\log\left( o(\mu^4) \right)}/{\log\left( {\left \|J_{\epsilon}^{\T} \right \|} \right)} \) and \( o(\mu^4) \) denotes a term that is higher in order than \( \mu^4 \).
\end{theorem}
\begin{proof}
	The argument is an adjustment of~\cite[Theorem 1]{Vlaski19nonconvexP1}.
\end{proof}
This result ensures that the entire network clusters around the network centroid \( \w_{c, i} \) after sufficient iterations, allowing us to leverage it as a proxy for all agents.
\begin{theorem}[\textbf{Descent relation}]\label{TH:DESCENT_RELATION}
	Beginning at \( \w_{c, i-1} \) in the large gradient regime \( \mathcal{G} \), we can bound:
  \begin{align}\label{eq:descent_in_g}
    &\:\E \left \{ J(\w_{c, i}) | \w_{c, i-1} \in \mathcal{G} \right \} \notag \\
    \le&\: \E \left \{ J(\w_{c, i-1}) | \w_{c, i-1} \in \mathcal{G} \right \} - \mu^2 \frac{c_2}{\pi} + \frac{O(\mu^3)}{\pi_{i-1}^{\mathcal{G}}}
  \end{align}
	{as long as \( \pi_{i-1}^{\mathcal{G}} = \mathrm{Pr}\left \{ \w_{c, i-1} \in \mathcal{G} \right \} \neq 0 \)} where the relevant constants are listed in definition~\ref{DEF:SETS}.
\end{theorem}
\begin{proof}
	The argument is an adjustment of~\cite[Theorem 2]{Vlaski19nonconvexP1}.
\end{proof}
\begin{theorem}[\textbf{Descent through strict saddle-points}]\label{TH:DESCENT_THROUGH_SADDLE_POINTS}
  {Suppose \( \pi_{i^{\star}}^{\mathcal{H}} \neq 0 \), i.e., \( \w_{c, i^{\star}} \)} is approximately stationary with significant negative eigenvalue. Then, iterating for \( i^s \) iterations after \( i^{\star} \) with
  \begin{align}
    i^{s} =  \frac{\log\left( 2 M  \frac{\left(\sum_{k=1}^K p_k^2 \sigma_k^2\right)}{\left(\sum_{k=1}^K p_k^2 \sigma_{k, \ell}^2\right)} + 1 \right)}{O(\mu \tau)}
  \end{align}
  guarantees
  \begin{align}
    &\: \E \left \{ J(\w_{c, i^{\star}+i^s}) | \w_{c, i^{\star}} \in \mathcal{H} \right \} \notag \\
    \le&\: \E \left \{ J(\w_{c, i^{\star}}) | \w_{c, i^{\star}} \in \mathcal{H} \right \} - \frac{\mu}{2} M \left(\sum_{k=1}^K p_k^2 \sigma_k^2\right) + \frac{o(\mu)}{\pi_{i^{\star}}^{\mathcal{H}}}
  \end{align}
\end{theorem}
\begin{proof}
  The argument is an adjustment of~\cite[Theorem 1]{Vlaski19nonconvexP2}.
\end{proof}
Theorem~\ref{TH:DESCENT_RELATION} ensures descent in one iteration as long as the gradient norm is sufficiently large, while~\ref{TH:DESCENT_THROUGH_SADDLE_POINTS} ensures descent even for first-order stationary points, as long as the Hessian has a negative eigenvalue in a number of iterations \( i^s \) that can be bounded. This ensures efficient escape from strict saddle-points. We conclude:
\begin{theorem}\label{TH:FINAL_THEOREM}
  For sufficiently small step-sizes \( \mu \), we have with probability \( 1 - \pi \), that \( \w_{c, i^o} \in \mathcal{M} \), i.e.,
  \begin{equation}\label{eq:first_guarantee}
    \| \nabla J(\w_{c, i^o}) \|^2 \le O\left( \mu \left( \sum_{k=1}^K p_k^2 \sigma_k^2\right) \right)
  \end{equation}
  and \( \lambda_{\min}\left( \nabla^2 J(\w_{c, i^o}) \right) \ge -\tau \) in at most \( i^o \) iterations, where
  \begin{align}\label{eq:conv_guarantee}
    i^o \le \frac{2 \left( J(w_{c, 0}) - J^o \right)}{\mu^2 \delta \left( \sum_{k=1}^K p_k^2 \sigma_k^2\right) \pi} i^s
  \end{align}
\end{theorem}
\begin{proof}
  The argument is an adjustment of~\cite[Theorem 2]{Vlaski19nonconvexP2}.
\end{proof}

\section{Comparative Analysis}
\subsection{Step-Size Normalization}
Note that in Theorem~\ref{TH:FINAL_THEOREM}, both the limiting accuracy~\eqref{eq:first_guarantee} and convergence rate~\eqref{eq:conv_guarantee} depend on the combination policy and network size through \( \sum_{k=1}^K p_k^2 \sigma_k^2 \). To facilitate comparison, we shall normalize the step-size in~\eqref{eq:adapt}:
\begin{equation}\label{eq:step_norm}
  \mu' \triangleq \frac{\mu}{\sum_{k=1}^K p_k^2 \sigma_k^2}
\end{equation}
Under this setting, Theorem~\ref{TH:FINAL_THEOREM} ensures a point \( i^o \) satisfying
\begin{equation}\label{eq:first_guarantee_norm}
  \| \nabla J(\w_{c, i^o}) \|^2 \le O\left( \mu \right)
\end{equation}
and \( \lambda_{\min}\left( \nabla^2 J(\w_{c, i^o}) \right) \ge -\tau \) in at most:
\begin{align}\label{eq:conv_guarantee_norm}
  i^o \le \frac{2 \left( J(w_{c, 0}) - J^o \right)}{\mu^2 \delta  \pi} \left( \sum_{k=1}^K p_k^2 \sigma_k^2\right) i^s
\end{align}
iterations with
\begin{equation}\label{eq:escape_time_norm}
  i^{s} = \frac{\log\left( 2 M  \frac{\left(\sum_{k=1}^K p_k^2 \sigma_k^2\right)}{\left(\sum_{k=1}^K p_k^2 \sigma_{k, \ell}^2\right)} + 1 \right)}{O(\mu \tau)} \left(\sum_{k=1}^K p_k^2 \sigma_k^2\right)
\end{equation}
Note that the normalization of the step-size causes~\eqref{eq:first_guarantee_norm} to become independent of \( \left(\sum_{k=1}^K p_k^2 \sigma_k^2\right) \), allowing for the fair evaluation of~\eqref{eq:conv_guarantee_norm} and~\eqref{eq:escape_time_norm} as a function of the number of agents.

\subsection{Linear Speedup Using Symmetric Combination Weights}
When the combination matrix \( A \) is symmetric, i.e., \( A = A^{\T} \), it follows that \( p_k = \frac{1}{K} \)~\cite{Sayed14}. For simplicity, in this section, we shall also assume a uniform data profile for all agents, i.e., that \( \sigma_k = \sigma \) and \( \sigma_{\ell, k} = \sigma_{\ell} \) for all \( k \). We obtain:
\begin{theorem}[\textbf{Linear Speedup for Symmetric Policies}]\label{TH:COR}
  Under the step-size normalization~\eqref{eq:step_norm}, and for symmetric combination policies \( A = A^{\T} \) with the uniform data profile \( \sigma_k^2 = \sigma^2 \) and \( \sigma_{\ell, k}^2 = \sigma_{\ell}^2 \) for all \( k \), the escape time simplifies to:
  \begin{equation}\label{eq:linear_escape}
    i^{s} = \frac{\log\left( 2 M  \frac{\sigma^2}{\sigma_{\ell}^2} + 1 \right)}{O(\mu \tau)} \frac{\sigma^2}{K} = O\left(\frac{1}{\mu \tau K} \right)
  \end{equation}
\end{theorem}
\begin{proof}
  The result follows immediately after cancellations.
\end{proof}

\subsection{Benefit of Employing Asymmetric Combination Weights}
In this subsection, we show how employing asymmetric combination weights can be beneficial in terms of the time required to escape saddle-points when the data profile across agents is no longer uniform. In particular, we will no longer require the upper and lower bounds \( \sigma_k^2 \) and \( \sigma_{\ell, k}^2 \) to be common for all agents, and no longer require the combination policy to be symmetric. Instead, to simplify the derivation, we assume that the gradient noise is approximately isotropic, i.e., \( \sigma_k^2 \approx \sigma_{k, \ell}^2 \) so that~\eqref{eq:escape_time_norm} can be simplified to:
\begin{align}\label{eq:approx_escape}
  i^{s} \approx O\left( \frac{\sum_{k=1}^K p_k^2 \sigma_k^2}{\mu \tau} \right)
\end{align}
Then, we can formulate the following optimization problem to minimize the escape time \( i^s \) over the space of valid combination policies:
\begin{align}
  \min_A \sum_{k=1}^K p_k^2 \sigma_k^2 \ \ \mathrm{s.t.}\ & a_{\ell k} \geq 0,\ \sum_{\ell \in \mathcal{N}_k} a_{\ell k}=1,\ a_{\ell k} = 0\ \mathrm{if}\ \ell \notin \mathcal{N}_k, \notag\\
  &Ap=p, \quad \mathds{1}^{\T} p=1, \quad p_k > 0.
\end{align}
This precise optimization problem has appeared before in the pursuit of asymmetric combination policies that minimize the steady-state error of the diffusion strategy~\eqref{eq:adapt}--\eqref{eq:combine} in \emph{strongly-convex} environments~\cite{Sayed14}. Its solution is available in closed form and can even be pursued in a decentralized manner, requiring only exchanges among neighbors~\cite{Sayed14}.
\begin{theorem}[\textbf{Metropolis-Hastings Combination Policy~\cite{Sayed14}}]\label{TH:METRO}
  Under the step-size normalization~\eqref{eq:step_norm}, the asymmetric Metropolis-Hastings combination policy minimizes the approximate saddle-point escape time~\eqref{eq:approx_escape}. It takes the form:
  \begin{equation}
    a_{\ell k}^o = \begin{cases} \frac{\sigma_k^2}{\max\{ n_k \sigma_k^2, n_{\ell} \sigma_{\ell}^2 \}}, \ \ \ &\ell \in \mathcal{N}_k, \\ 1 - \sum_{m \in \mathcal{N}_k \backslash \{k\}} a_{m k}^o, \ \ \ &\ell = k. \end{cases}
  \end{equation}
  where \( n_k = | \mathcal{N}_k | \) denotes the size of the neighborhood of agent \( k \).
\end{theorem}

\section{Simulations}
We construct a sample landscape to verify the linear speedup in the size of the network indicated by the analysis in this work. The loss function is constructed from a single-layer neural network with a linear hidden layer and a logistic activation function for the output layer. Penalizing this architecture with the cross-entropy loss gives:
\begin{align}
  J(w_1, W_2) = \E \log\left({1+e^{- \boldsymbol{\gamma} w_1^{\T} W_2 \h}}\right) + \frac{\rho}{2}\|w_1\|^2 + \frac{\rho}{2} \| W_2 \|_F^2
\end{align}
where \( w_1 \) and \( W_2 \) denote the weights of the individual layers, \( \h \in \mathds{R}^M \) denotes the feature vector, and \( \boldsymbol{\gamma} \in \left \{ \pm 1 \right \} \) is the class variable. It can be verified that this loss has a single strict saddle-point at \( w_1 = W_2 = 0 \) and global minima in the positive and negative quadrant, respectively~\cite{Vlaski19nonconvexP2}. We show the evolution of the function value at the network centroid under the step-size normalization rule~\eqref{eq:step_norm} and observe a linear speedup in \( K \), consistent with~\eqref{eq:linear_escape} while noting no significant differences in steady-state performance, which is consistent with~\eqref{eq:first_guarantee_norm}.
\begin{figure}[htb]
  \vspace{-3.5mm}
	\centering
	\includegraphics[width=.95\columnwidth]{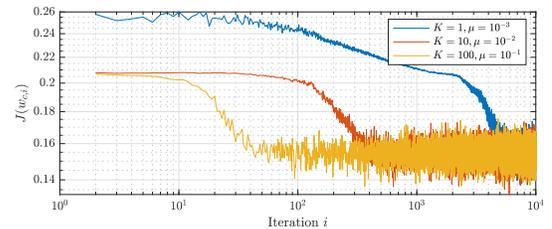}\vspace{-2mm}
\caption{Linear speedup in saddle-point escape time.}\label{fig:performance}
\end{figure}

\clearpage
\bibliographystyle{IEEEbib}
{\bibliography{main}}

\end{document}